\documentclass{amsart}
\usepackage{amsfonts}
\usepackage{amsmath,amscd}
\usepackage{amsthm}
\usepackage{amssymb}
\usepackage{latexsym}

\newtheorem{prop}{Proposition}[section]

\newtheorem{defn}{Definition}[section]

\newtheorem{rem}{Remark}

\numberwithin{equation}{section}
\begin{document}
\newcommand{\beqa}{\begin{eqnarray}}
\newcommand{\eeqa}{\end{eqnarray}}
\newcommand{\thmref}[1]{Theorem~\ref{#1}}
\newcommand{\secref}[1]{Sect.~\ref{#1}}
\newcommand{\lemref}[1]{Lemma~\ref{#1}}
\newcommand{\propref}[1]{Proposition~\ref{#1}}
\newcommand{\corref}[1]{Corollary~\ref{#1}}
\newcommand{\remref}[1]{Remark~\ref{#1}}
\newcommand{\er}[1]{(\ref{#1})}
\newcommand{\nc}{\newcommand}
\newcommand{\rnc}{\renewcommand}

\nc{\cal}{\mathcal}

\nc{\goth}{\mathfrak}
\rnc{\bold}{\mathbf}
\renewcommand{\frak}{\mathfrak}
\renewcommand{\Bbb}{\mathbb}

\newcommand{\id}{\text{id}}
\nc{\Cal}{\mathcal}
\nc{\Xp}[1]{X^+(#1)}
\nc{\Xm}[1]{X^-(#1)}
\nc{\on}{\operatorname}
\nc{\ch}{\mbox{ch}}
\nc{\Z}{{\bold Z}}
\nc{\J}{{\mathcal J}}
\nc{\C}{{\bold C}}
\nc{\Q}{{\bold Q}}
\renewcommand{\P}{{\mathcal P}}
\nc{\N}{{\Bbb N}}
\nc\beq{\begin{equation}}
\nc\enq{\end{equation}}
\nc\lan{\langle}
\nc\ran{\rangle}
\nc\bsl{\backslash}
\nc\mto{\mapsto}
\nc\lra{\leftrightarrow}
\nc\hra{\hookrightarrow}
\nc\sm{\smallmatrix}
\nc\esm{\endsmallmatrix}
\nc\sub{\subset}
\nc\ti{\tilde}
\nc\nl{\newline}
\nc\fra{\frac}
\nc\und{\underline}
\nc\ov{\overline}
\nc\ot{\otimes}
\nc\bbq{\bar{\bq}_l}
\nc\bcc{\thickfracwithdelims[]\thickness0}
\nc\ad{\text{\rm ad}}
\nc\Ad{\text{\rm Ad}}
\nc\Hom{\text{\rm Hom}}
\nc\End{\text{\rm End}}
\nc\Ind{\text{\rm Ind}}
\nc\Res{\text{\rm Res}}
\nc\Ker{\text{\rm Ker}}
\rnc\Im{\text{Im}}
\nc\sgn{\text{\rm sgn}}
\nc\tr{\text{\rm tr}}
\nc\Tr{\text{\rm Tr}}
\nc\supp{\text{\rm supp}}
\nc\card{\text{\rm card}}
\nc\bst{{}^\bigstar\!}
\nc\he{\heartsuit}
\nc\clu{\clubsuit}
\nc\spa{\spadesuit}
\nc\di{\diamond}
\nc\cW{\cal W}
\nc\cG{\cal G}
\nc\al{\alpha}
\nc\bet{\beta}
\nc\ga{\gamma}
\nc\de{\delta}
\nc\ep{\epsilon}
\nc\io{\iota}
\nc\om{\omega}
\nc\si{\sigma}
\rnc\th{\theta}
\nc\ka{\kappa}
\nc\la{\lambda}
\nc\ze{\zeta}

\nc\vp{\varpi}
\nc\vt{\vartheta}
\nc\vr{\varrho}

\nc\Ga{\Gamma}
\nc\De{\Delta}
\nc\Om{\Omega}
\nc\Si{\Sigma}
\nc\Th{\Theta}
\nc\La{\Lambda}

\nc\boa{\bold a}
\nc\bob{\bold b}
\nc\boc{\bold c}
\nc\bod{\bold d}
\nc\boe{\bold e}
\nc\bof{\bold f}
\nc\bog{\bold g}
\nc\boh{\bold h}
\nc\boi{\bold i}
\nc\boj{\bold j}
\nc\bok{\bold k}
\nc\bol{\bold l}
\nc\bom{\bold m}
\nc\bon{\bold n}
\nc\boo{\bold o}
\nc\bop{\bold p}
\nc\boq{\bold q}
\nc\bor{\bold r}
\nc\bos{\bold s}
\nc\bou{\bold u}
\nc\bov{\bold v}
\nc\bow{\bold w}
\nc\boz{\bold z}

\nc\ba{\bold A}
\nc\bb{\bold B}
\nc\bc{\bold C}
\nc\bd{\bold D}
\nc\be{\bold E}
\nc\bg{\bold G}
\nc\bh{\bold H}
\nc\bi{\bold I}
\nc\bj{\bold J}
\nc\bk{\bold K}
\nc\bl{\bold L}
\nc\bm{\bold M}
\nc\bn{\bold N}
\nc\bo{\bold O}
\nc\bp{\bold P}
\nc\bq{\bold Q}
\nc\br{\bold R}
\nc\bs{\bold S}
\nc\bt{\bold T}
\nc\bu{\bold U}
\nc\bv{\bold V}
\nc\bw{\bold W}
\nc\bz{\bold Z}
\nc\bx{\bold X}

\nc\ca{\mathcal A}
\nc\cb{\mathcal B}
\nc\cc{\mathcal C}
\nc\cd{\mathcal D}
\nc\ce{\mathcal E}
\nc\cf{\mathcal F}
\nc\cg{\mathcal G}
\rnc\ch{\mathcal H}
\nc\ci{\mathcal I}
\nc\cj{\mathcal J}
\nc\ck{\mathcal K}
\nc\cl{\mathcal L}
\nc\cm{\mathcal M}
\nc\cn{\mathcal N}
\nc\co{\mathcal O}
\nc\cp{\mathcal P}
\nc\cq{\mathcal Q}
\nc\car{\mathcal R}
\nc\cs{\mathcal S}
\nc\ct{\mathcal T}
\nc\cu{\mathcal U}
\nc\cv{\mathcal V}
\nc\cz{\mathcal Z}
\nc\cx{\mathcal X}
\nc\cy{\mathcal Y}

\nc\e[1]{E_{#1}}
\nc\ei[1]{E_{\delta - \alpha_{#1}}}
\nc\esi[1]{E_{s \delta - \alpha_{#1}}}
\nc\eri[1]{E_{r \delta - \alpha_{#1}}}
\nc\ed[2][]{E_{#1 \delta,#2}}
\nc\ekd[1]{E_{k \delta,#1}}
\nc\emd[1]{E_{m \delta,#1}}
\nc\erd[1]{E_{r \delta,#1}}

\nc\ef[1]{F_{#1}}
\nc\efi[1]{F_{\delta - \alpha_{#1}}}
\nc\efsi[1]{F_{s \delta - \alpha_{#1}}}
\nc\efri[1]{F_{r \delta - \alpha_{#1}}}
\nc\efd[2][]{F_{#1 \delta,#2}}
\nc\efkd[1]{F_{k \delta,#1}}
\nc\efmd[1]{F_{m \delta,#1}}
\nc\efrd[1]{F_{r \delta,#1}}

\nc\fa{\frak a}
\nc\fb{\frak b}
\nc\fc{\frak c}
\nc\fd{\frak d}
\nc\fe{\frak e}
\nc\ff{\frak f}
\nc\fg{\frak g}
\nc\fh{\frak h}
\nc\fj{\frak j}
\nc\fk{\frak k}
\nc\fl{\frak l}
\nc\fm{\frak m}
\nc\fn{\frak n}
\nc\fo{\frak o}
\nc\fp{\frak p}
\nc\fq{\frak q}
\nc\fr{\frak r}
\nc\fs{\frak s}
\nc\ft{\frak t}
\nc\fu{\frak u}
\nc\fv{\frak v}
\nc\fz{\frak z}
\nc\fx{\frak x}
\nc\fy{\frak y}

\nc\fA{\frak A}
\nc\fB{\frak B}
\nc\fC{\frak C}
\nc\fD{\frak D}
\nc\fE{\frak E}
\nc\fF{\frak F}
\nc\fG{\frak G}
\nc\fH{\frak H}
\nc\fJ{\frak J}
\nc\fK{\frak K}
\nc\fL{\frak L}
\nc\fM{\frak M}
\nc\fN{\frak N}
\nc\fO{\frak O}
\nc\fP{\frak P}
\nc\fQ{\frak Q}
\nc\fR{\frak R}
\nc\fS{\frak S}
\nc\fT{\frak T}
\nc\fU{\frak U}
\nc\fV{\frak V}
\nc\fZ{\frak Z}
\nc\fX{\frak X}
\nc\fY{\frak Y}
\nc\tfi{\ti{\Phi}}
\nc\bF{\bold F}
\rnc\bol{\bold 1}

\nc\ua{\bold U_\A}

\nc\qinti[1]{[#1]_i}
\nc\q[1]{[#1]_q}
\nc\xpm[2]{E_{#2 \delta \pm \alpha_#1}}  
\nc\xmp[2]{E_{#2 \delta \mp \alpha_#1}}
\nc\xp[2]{E_{#2 \delta + \alpha_{#1}}}
\nc\xm[2]{E_{#2 \delta - \alpha_{#1}}}
\nc\hik{\ed{k}{i}}
\nc\hjl{\ed{l}{j}}
\nc\qcoeff[3]{\left[ \begin{smallmatrix} {#1}& \\ {#2}& \end{smallmatrix}
\negthickspace \right]_{#3}}
\nc\qi{q}
\nc\qj{q}

\nc\ufdm{{_\ca\bu}_{\rm fd}^{\le 0}}


\nc\isom{\cong} 

\nc{\pone}{{\Bbb C}{\Bbb P}^1}
\nc{\pa}{\partial}
\def\H{\mathcal H}
\def\L{\mathcal L}
\nc{\F}{{\mathcal F}}
\nc{\Sym}{{\goth S}}
\nc{\A}{{\mathcal A}}
\nc{\arr}{\rightarrow}
\nc{\larr}{\longrightarrow}

\nc{\ri}{\rangle}
\nc{\lef}{\langle}
\nc{\W}{{\mathcal W}}
\nc{\uqatwoatone}{{U_{q,1}}(\su)}
\nc{\uqtwo}{U_q(\goth{sl}_2)}
\nc{\dij}{\delta_{ij}}
\nc{\divei}{E_{\alpha_i}^{(n)}}
\nc{\divfi}{F_{\alpha_i}^{(n)}}
\nc{\Lzero}{\Lambda_0}
\nc{\Lone}{\Lambda_1}
\nc{\ve}{\varepsilon}
\nc{\phioneminusi}{\Phi^{(1-i,i)}}
\nc{\phioneminusistar}{\Phi^{* (1-i,i)}}
\nc{\phii}{\Phi^{(i,1-i)}}
\nc{\Li}{\Lambda_i}
\nc{\Loneminusi}{\Lambda_{1-i}}
\nc{\vtimesz}{v_\ve \otimes z^m}

\nc{\asltwo}{\widehat{\goth{sl}_2}}
\nc\ag{\widehat{\goth{g}}}  
\nc\teb{\tilde E_\boc}
\nc\tebp{\tilde E_{\boc'}}

\title[Generalized $q-$Onsager algebras and boundary Toda field theories]{Generalized $q-$Onsager algebras and\\ boundary affine Toda field theories}
\author{P. Baseilhac}
\address{Laboratoire de Math\'ematiques et Physique Th\'eorique CNRS/UMR 6083,
         F\'ed\'eration Denis Poisson, Universit\'e de Tours, Parc de Grammont, 37200 Tours, FRANCE}
         \email{baseilha@lmpt.univ-tours.fr}

\author{S. Belliard}
\address{Istituto Nazionale di Fisica Nucleare, Sezione di Bologna, Via Irnerio 46,  40126 Bologna,  Italy}
\email{belliard@bo.infn.it}

\begin{abstract}
Generalizations of the $q-$Onsager algebra are introduced and studied. In one of the simplest case and $q=1$, the algebra reduces to the one proposed by Uglov-Ivanov. In the general case and $q\neq 1$, an explicit algebra homomorphism associated with coideal subalgebras of quantum affine Lie algebras (simply and non-simply laced) is exhibited. Boundary (soliton non-preserving) integrable quantum Toda field theories are then considered in light of these results. For the first time, all defining relations for the underlying non-Abelian symmetry algebra are explicitely obtained. As a consequence, based on purely algebraic arguments all integrable (fixed or dynamical) boundary conditions are classified.    
\end{abstract}

\maketitle

\vskip -0.2cm



\section{Introduction}
In recent years, a new algebraic structure called the $q-$Onsager algebra (or equivalently the tridiagonal algebra) has emerged in different problems of mathematical physics.

On one side, it appears in the mathematical litterature of $P-$ and $Q-$polynomial association schemes  and their relationship with the Askey scheme of orthogonal polynomials \cite{Zhed,Ter01,GruHa,Ter1,Ter2}, related Jacobi matrices and, more generally, certain families of symmetric functions of one variable and related block tridiagonal matrices (see e.g. \cite{Ter3,TDpair}). 

On the other side, this algebra appears in several quantum integrable systems. Playing a crucial role at $q=1$ in the exact solution of the planar Ising \cite{Ons44} and superintegrable Potts model \cite{VoGR}, it also finds applications in solving the XXZ open spin chain with non-diagonal boundary parameters and generic deformation parameter $q$. Indeed, the transfer matrix of this model has been shown to admit an expansion in terms of the elements of the $q-$Onsager algebra \cite{BasK0,BasK1} acting on some finite dimensional representation. As a consequence, the solution of the model i.e. the complete spectrum and eigenstates can be derived using solely its representation theory, bypassing the Bethe ansatz approach which does not apply in the generic regime of parameters \cite{BasK2}. Appart from lattice models, in quantum field theory the $q-$Onsager algebra is known to be the hidden non-Abelian symmetry of the boundary sine-Gordon model \cite{B1,B2}. \vspace{1mm} 

By definition, the $q-$Onsager algebra is an associative algebra with unity generated by two elements (called
the standard generators), say $\textsf{A}_0,\textsf{A}_1$. Introducing the $q-$commutator\,\footnote{For further convenience, definitions for the parameter $q$ and the $q-$commutator chosen here differ compared to \cite{TDpair,BasK0,BasK1,BasK2}.}
$\big[X,Y\big]_q=XY-qYX$, the fundamental (sometimes called $q-$Dolan-Grady) relations take the form
\beqa
[\textsf{A}_0,[\textsf{A}_0,[\textsf{A}_0,\textsf{A}_1]_{q^2}]_{q^{-2}}]=\rho_0[\textsf{A}_0,\textsf{A}_1]\
,\qquad
[\textsf{A}_1,[\textsf{A}_1,[\textsf{A}_1,\textsf{A}_0]_{q^2}]_{q^{-2}}]=\rho_1[\textsf{A}_1,\textsf{A}_0]\
\label{Talg}\eeqa
where $q$ is a deformation parameter (assumed to be not a root of
unity) and $\rho_0,\rho_1$ are fixed scalars.
Note that for $\rho_0=\rho_1=0$ this algebra
reduces to the $q-$Serre relations of $U_q(\widehat{sl_2})$, and for $q=1$, $\rho_0=\rho_1=16$ it leads
to the Onsager algebra \cite{Ons44,Pe} defined by the Dolan-Grady relations \cite{DG}.\vspace{1mm}

Similarly to the well-established relationship between the Onsager algebra and the affine Lie algebra $\widehat{sl_2}$ \cite{Da,DateRoan}, the $q-$Onsager algebra (\ref{Talg}) is actually closely related with the $U_q(\widehat{sl_2})$ algebra, a fact that may be also expected from the structure of the l.h.s. of (\ref{Talg}) compared with the $q-$Serre relations of $U_q(\widehat{sl_2})$. Indeed, examples of algebra homomorphisms for the standard generators $\textsf{A}_0,\textsf{A}_1$ have been proposed for $\rho_0\neq 0,\rho_1\neq 0$, and related finite dimensional representations studied in details. We refer the reader to \cite{ITer,B2,AlCu,ITer2} for details. In particular, the following realization immediately follows from \cite{B2}:
\beqa
{\textsf A}_0&=& c_0e_0q^{h_0/2} + \overline{c}_0f_0q^{h_0/2} + \epsilon_0q^{h_0}\ ,\nonumber\\
{\textsf A}_1&=& c_1e_1q^{h_1/2} + \overline{c}_1f_1q^{h_1/2} + \epsilon_1q^{h_1} \ ,\label{realsl2}
\eeqa
where\,\footnote{Defining relations of $U_q(\widehat{sl_2})$ are given in the next section.} $\{h_i,e_i,f_i\}$ denote the generators of $U_q(\widehat{sl_2})$ and
one identifies $\rho_i=c_i\overline{c}_i(q+q^{-1})^2$ for $i=0,1$. Thanks to the Hopf algebra structure of $U_q(\widehat{sl_2})$, finite dimensional representations have been studied in details (see for instance \cite{TDpair,ITer2}). In addition, a new type of current algebra has been recently derived \cite{BasS} which rigorously establishes the isomorphism between the reflection equation algebra associated with $U_q(\widehat{sl_2})$ $R-$matrices and the $q-$Onsager algebra (\ref{Talg}).\vspace{1mm}

In the context of quantum integrable systems, the elements ${\textsf A}_0,{\textsf A}_1$ take the form of non local operators on the lattice or continuum. According to the model and objective considered, they are used either to eventually derive second order difference equations fixing the spectrum of the model \cite{BasK2}, or the complete set of scattering amplitudes of the fundamental particles \cite{MN98,DM,BasK3}.\vspace{1mm}

In view of all these results, finding an analogue of the deformed relations (\ref{Talg}) that may be related to {\it higher rank} affine Lie algebras in a similar manner, as well as considering potential implications for quantum integrable systems with extended symmetries seems to be a rather interesting problem. In the undeformed case $q=1$, a step towards this direction has been made by Uglov and Ivanov who introduced the so-called $sl_n-$Onsager's algebra for $n\geq 2$. However, to our knowledge since these results no further progress in this direction were ever published.\vspace{1mm}

In the present letter, we remedy this situation. Namely, to each affine Lie algebra (of classical or exceptional type) ${\widehat g}$ we associate a $q-$Onsager algebra denoted $O_q({\widehat g})$. Then, by analogy with the ${\widehat{sl_2}}$ case, we propose an algebra homomorphism from $O_q({\widehat g})$ to the coideal subalgebra\,\footnote{For definitions, see e.g. \cite{Mol,Letzter}} of $U_q({\widehat g})$ generalizing (\ref{realsl2}). Applications to boundary quantum affine Toda field theories introduced in \cite{Fring,Cor} - with soliton non-preserving boundary conditions - are then considered. Despite of the fact that defining relations of the underlying hidden symmetry in these models were {\it not} known up to now (except for the sine-Gordon model \cite{B1,B2}), the explicit knowledge of non-local conserved charges have provided a powerful tool to construct boundary reflection matrices at least for ${\widehat g}\equiv a_n^{(1)},d_n^{(1)}$ cases \cite{MN98,DM,DelG}. Here and for the first time, we show that each boundary affine Toda field theory of the family defined in \cite{Fring,Cor} associated with ${\widehat g}$ enjoys a hidden non-Abelian symmetry of type $O_q({\widehat g})$. As a consequence, all known scalar integrable boundary conditions \cite{Cor} simply follow from the algebraic structure, with no reference to its representation theory\,\footnote{Contrary to previous works, which are representation's dependent.}. More generally, all possible integrable dynamical boundary conditions (additional degrees of freedom are located at the boundary) admissible in these models are also classified according to this new framework, generalizing the results of the boundary sine-Gordon model with dynamical boundary conditions \cite{BasDel,BasK3}.\vspace{2mm}

\section{Generalizations of the $q-$Onsager algebra}
As mentionned in the introduction, generalized $q-$Onsager algebras can be introduced by analogy with (\ref{Talg}). Having in mind the structure of $q-$Serre relations for higher rank affine Lie algebras and their potential relations with coideal subalgebras of quantum affine algebras, a general formulation can be proposed.
\begin{defn}{\label{defqOns}} Let $\{a_{ij}\}$ be the extended Cartan matrix of the affine Lie algebra $\widehat{g}$ with Dynkin diagram reported in Appendix A.
Fix coprime integers $d_i$ such that $d_ia_{ij}$ is symmetric.
The generalized $q-$Onsager algebra $O_q(\widehat{g})$ is an associative algebra
with unit $1$, elements $\textsf{A}_i$ and scalars $\rho^{k}_{ij},\gamma^{kl}_{ij}\in {\mathbb C}$ with $i,j \in \{0,1,...,n\}$, $k \in \{0,1,...,[-\frac{a_{ij}}{2}]-1\}\,\footnote{$[a]$ means the nearest  higher integer of $a$ with [1/2]=1.} $ and $l \in \{0,1,...,-a_{ij}-1-2\,k\}$ ($k$ and $l$ are positive integer). 
The defining relations are :
\begin{align}
\sum_{r=0}^{1-a_{ij}}(-1)^r
\left[ \begin{array}{c}
1-a_{ij} \\
r 
\end{array}\right]_{q_{i}}
\textsf{A}_i^{1-a_{ij}-r}\textsf{A}_j\,\textsf{A}_i^r=
\sum_{k=0}^{[-\frac{a_{ij}}{2}]-1}\rho^{k}_{ij}\sum_{l=0}^{-2\,k-a_{ij}-1}\,(-1)^l\,\gamma^{kl}_{ij}\,
\textsf{A}_i^{-2\,k-a_{ij}-1-l}\textsf{A}_j\,\textsf{A}_i^l\ ,\label{genqOns}
\end{align}
where the constants $\gamma^{kl}_{ij}$ are such that:
\begin{align}
&\qquad\qquad\qquad\mbox{For} \qquad a_{ij}=a_{ji}=-1\quad : \quad \gamma^{00}_{ij}=\gamma^{00}_{ji}=1\ ;\qquad\qquad\qquad\qquad\qquad\qquad\qquad
\qquad  \qquad  \qquad  \qquad  \qquad  \qquad   \nonumber\\
&\qquad\qquad\qquad\mbox{For} \qquad a_{ij}=-1  \qquad  \mbox{and}  \qquad a_{ji}=-2\ \ : \quad \gamma^{00}_{ij}=\gamma^{00}_{ji}=\gamma^{01}_{ji}=1\ ; \qquad\qquad\qquad\qquad\qquad\quad\qquad\quad \quad  \quad  \quad \nonumber\\
&\qquad\qquad\qquad\mbox{For} \qquad a_{ij}=-1  \qquad  \mbox{and}  \qquad a_{ji}=-3 \quad : \quad 
\gamma^{00}_{ij}=1
, \qquad  \gamma^{00}_{ji}= \gamma^{02}_{ji}=\gamma^{10}_{ji}=1\ ,\qquad \qquad\qquad\qquad\qquad\quad \ \!\nonumber\\
&\qquad\qquad\qquad\qquad\qquad\qquad\qquad\qquad\qquad\qquad\qquad\qquad\qquad\gamma^{01}_{ji}= \frac{(q+q^{-1})(q^2+q^{-2})(q^2 + 3 + q^{-2})}
{(q^4 + 2 q^2 + 4 + 2 q^{-2} + q^{-4})}\ ;\qquad\nonumber\\
&\qquad\qquad\qquad\mbox{For} \qquad a_{ij}=-1  \qquad  \mbox{and}  \qquad a_{ji}=-4 \quad : \quad 
\gamma^{00}_{ij}=1,
\qquad  \gamma^{00}_{ji}= \gamma^{03}_{ji}= \gamma^{1l}_{ji}=1\ ,\qquad \ \
\nonumber\\
&\qquad\qquad\qquad\qquad\qquad\qquad\qquad\qquad\qquad\qquad\qquad\qquad\qquad\gamma^{01}_{ji}=\gamma^{02}_{ji}=\frac{[3]_q[5]_q}{q^4+q^{-4}+3} \ .\qquad  \qquad \qquad \qquad
\nonumber
\end{align}
\end{defn}
\vspace{-3mm}
\begin{rem}
For ${\widehat g}\equiv a_n^{(1)}$, $q=1$ and $\rho^0_{ij}=1$, the relations reduce to the ones of Uglov-Ivanov's $sl_n-$Onsager's algebra \cite{Uglov}. For ${\widehat g}\equiv a_n^{(1)}$ and $q\neq 1$, the relations already appeared in \cite{B1} without detailed explanations. For simply laced cases, note the close relationship with the defining relations of coideal subalgebras or the non-standard deformation of finite dimensional Lie algebras \cite{Letzter,Gavr,Klim}. 
\end{rem} 
\vspace{2mm}

For $q\neq 1$, an explicit relationship with coideal subalgebras of $U_q({\widehat g})$ can be easily exhibited. To this end, let us first recall some definitions that will be useful below. Define for $q\in {\mathbb C}^*$
\begin{align}
\left[ \begin{array}{c}
a \\
b 
\end{array}\right]_q
=\frac{[a]_q!}{[b]_q!\,[a-b]_q!}\ , \qquad
[a]_q!=[a]_q\,[a-1]_q \dots [1]_q\ ,\qquad
[a]_q=\frac{q^a-q^{-a}}{q-q^{-1}}, \quad [0]_q=1 \ .\nonumber
\end{align}

\begin{defn}{\label{defUq}}\cite{J1} Let $\{a_{ij}\}$ be the extended Cartan matrix of the affine Lie algebra $\widehat{g}$ with Dynkin diagram given in Appendix A. Fix coprime integers $d_i$ such that $d_ia_{ij}$ is symmetric.
$U_q(\widehat{g})$ is an associative algebra over ${\mathbb C}$ with unit $1$ generated by the elements
$\{ e_i, f_i,q_i^{\pm \frac{h_i}{2}}\}$, $i \in 0 \dots n$ subject to the relations:

\begin{align}
q_i^{\pm\frac{h_i}{2}}q_i^{\mp \frac{h_i}{2}}=1,
\qquad q_i^{\frac{h_i}{2}}q_j^{\frac{h_j}{2}}=q_j^{\frac{h_j}{2}}q_i^{\frac{h_i}{2}}\ , \qquad \qquad \qquad \qquad \qquad \qquad \qquad \qquad \qquad \quad \nonumber\\
q_i^{\frac{h_i}{2}}\,e_j\,q_i^{-\frac{h_i}{2}}=q_i^{\frac{a_{ij}}{2}}\,e_j,
\qquad q_i^{\frac{h_i}{2}}\,f_j\,q_i^{-\frac{h_i}{2}}=q_i^{-\frac{a_{ij}}{2}}\,f_j\ ,\qquad
\null [e_i,f_j]=\delta_{ij}\frac{q_i^{h_i}-q_i^{-h_i}}{q_i-q_i^{-1}}\ ,\nonumber\\
e_i\,e_j=e_j\,e_i\ , \qquad f_i\,f_j=f_j\,f_i\ , \,\qquad \mbox{for} \qquad |i-j|>1\ ,\qquad \qquad \qquad \qquad \qquad \quad \nonumber\\
\sum_{r=0}^{1-a_{ij}}(-1)^r
\left[ \begin{array}{c}
1-a_{ij} \\
r 
\end{array}\right]_{q_i}
e_i^{1-a_{ij}-r}\,e_j\,e_i^r=0\ , \qquad \qquad \qquad \qquad \qquad \qquad \qquad \qquad \quad \nonumber\\
\sum_{r=0}^{1-a_{ij}}(-1)^r
\left[ \begin{array}{c}
1-a_{ij} \\
r 
\end{array}\right]_{q_i}
f_i^{1-a_{ij}-r}\,f_j\,\,f_i^r=0\ .\quad \ \ \qquad \qquad \qquad \qquad \qquad \qquad \qquad \quad \nonumber
\end{align}
\vspace{1mm}
The Hopf algebra structure is ensured by the existence of a
comultiplication $\Delta: U_{q}(\widehat{g})\mapsto U_{q}(\widehat{g})\otimes U_{q}(\widehat{g})$, antipode ${\cal S}: U_{q}(\widehat{g})\mapsto U_{q}(\widehat{g})$ 
and a counit ${\cal E}: U_{q}(\widehat{g})\mapsto {\mathbb C}$ with
\beqa \Delta(e_i)&=&e_i\otimes q_i^{-h_i/2} +
q_i^{h_i/2}\otimes e_i\ ,\nonumber \\
\Delta(f_i)&=&f_i\otimes q_i^{-h_i/2} + q_i^{h_i/2}\otimes f_i\ ,\nonumber\\
\Delta(h_i)&=&h_i\otimes I\!\!I + I\!\!I \otimes h_i\ ,\label{coprod}
\eeqa
\beqa {\cal S}(e_i)=-e_iq_i^{-h_i}\ ,\quad {\cal S}(f_i)=-q_i^{h_i}f_i\ ,\quad {\cal S}(h_i)=-h_i \ ,\qquad {\cal S}({I\!\!I})=1\
\label{antipode}\nonumber\eeqa
and\vspace{-0.3cm}
\beqa {\cal E}(e_i)={\cal E}(f_i)={\cal
E}(h_i)=0\ ,\qquad {\cal E}({I\!\!I})=1\
.\label{counit}\nonumber\eeqa
\end{defn}
\vspace{3mm}

Based on the realization (\ref{realsl2}) of the algebra (\ref{Talg}) for the simplest case ${\widehat{sl_2}}\equiv a_1^{(1)}$, and the results in \cite{DM,DelG} it looks rather natural to consider the following realizations for the generalized $q-$Onsager algebras.
\begin{prop}
Let $\{c_i,\overline{c}_i\}\in {\mathbb C}$ and $\{w_i\}\in{\mathbb C}^*$. There is an algebra homomorphism $\ O_q(\widehat{g})\ \rightarrow U_q(\widehat{g})$ such that
\begin{align}
\textsf{A}_i=c_i\,e_iq_i^{\frac{h_i}{2}} +\overline{c}_i\,f_iq_i^{\frac{h_i}{2}} + w_i q_i^{h_i}\ \label{realg}
\end{align}
iff the parameters $w_i$ are subject to the following constraints: 
\begin{align}
\mbox{For}\qquad {\widehat{g}}=a_n^{(1)} (n>1), d_n^{(1)}, e_6^{(1)},e_7^{(1)},e_8^{(1)}: 
\qquad  \Big\{\begin{array}{c} w_i\,\Big(w_j^2+\frac{c_j\,\overline{c}_j}{q+q^{-1}-2}\Big)=0 \\
w_j\,\Big(w_i^2+\frac{c_i\,\overline{c}_i}{q+q^{-1}-2}\Big)=0
\end{array}
\qquad \mbox{where} \quad i,j \ \mbox{are simply linked}\ .\quad \nonumber
\end{align}
\begin{align}
\mbox{For}\qquad {\widehat{g}}=b_n^{(1)}, c_n^{(1)}, a_{2n}^{(2)},a_{2n-1}^{(2)},d_{n+1}^{(2)},, e_6^{(2)},f_4^{(1)}: \qquad \qquad  \qquad \qquad  \qquad \qquad   \qquad \qquad\qquad \qquad \qquad \qquad \nonumber\\
\qquad   w_j\,\Big(w_i^2+\frac{c_i\,\overline{c}_i}{q_i+q_i^{-1}-2}\Big)=0 \
\qquad \mbox{if} \quad i,j \ \mbox{are doubly linked with $i$ the longest root}\ ;\nonumber\\
\qquad  \Big\{\begin{array}{c} w_i\,\Big(w_j^2+\frac{c_j\,\overline{c}_j}{q_j+q_j^{-1}-2}\Big)=0 \\
w_j\,\Big(w_i^2+\frac{c_i\,\overline{c}_i}{q_i+q_i^{-1}-2}\Big)=0
\end{array}
\qquad \mbox{if} \quad i,j \ \mbox{are simply linked}\ . \qquad  \qquad \qquad  \qquad \qquad\nonumber
\end{align}
\begin{align}
\mbox{For}\qquad {\widehat{g}}=g_2^{(1)}, d_{4}^{(3)}: \qquad \qquad\qquad\qquad \qquad \qquad\qquad\qquad\qquad\qquad \qquad\quad \qquad\ \qquad\qquad\qquad\qquad\qquad\qquad\qquad\qquad\qquad\nonumber\\
\qquad  \Big\{\begin{array}{c} w_j\Big(w_i^2+\frac{c_i\,\overline{c}_i}{(q_i+q_i^{-1}-2)}\Big)=0 \\
w_i\Big(w_j^2+\frac{c_j\,\overline{c}_j}{(q_j+q_j^{-1}-2)}\Big)\Big(w_j^2+\frac{c_j\,\overline{c}_j(q_j+q_j^{-1}-1)^2}{(q_j+q_j^{-1}-2)}\Big)=0
\end{array}
\qquad \mbox{if} \quad i,j \ \mbox{are triply linked with $i$ the longest root}\ .\qquad \qquad\quad\nonumber\\
\qquad  \Big\{\begin{array}{c} w_i\,\Big(w_j^2+\frac{c_j\,\overline{c}_j}{q_j+q_j^{-1}-2}\Big)=0 \\
w_j\,\Big(w_i^2+\frac{c_i\,\overline{c}_i}{q_i+q_i^{-1}-2}\Big)=0
\end{array}
\qquad \mbox{if} \quad i,j \ \mbox{are simply linked}\ . \qquad\qquad\qquad\qquad\qquad\qquad \qquad\nonumber
\end{align}
\begin{align}
\mbox{For}\qquad {\widehat{g}}=a_2^{(2)}:\qquad
w_j\,\Big(w_i^2+\frac{c_i\,\overline{c}_i}{(q_i+q_i^{-1}-2)}\Big)=0 \ \qquad \quad \mbox{with $i$ the longest root}\ .\qquad\qquad\qquad\qquad\qquad\qquad\qquad\qquad\qquad\qquad\qquad\nonumber
\end{align}
\end{prop}
\vspace{3mm}
\begin{proof} Plugging (\ref{realg}) into the relations of Definition 2.1, straightforward calculations leave few unwanted terms that cancel provided the above constraints on parameters $w_i$ are satisfied. 
The structure constants $\rho^k_{ij}$ - with respect to the indices $i,j$ - are identified as follows:
\begin{align}
&\qquad\mbox{For}\qquad {\widehat{g}}=a_n^{(1)} (n>1), d_n^{(1)}, e_6^{(1)},e_7^{(1)},e_8^{(1)}: \qquad \rho^0_{ij}=c_i\,\overline{c}_i \quad \mbox{and}\quad \rho^0_{ji}=c_j\,
\overline{c}_j\ .\qquad\qquad \nonumber\\
&\qquad\mbox{For}\qquad {\widehat{g}}=b_n^{(1)}, c_n^{(1)}, a_{2n}^{(2)},a_{2n-1}^{(2)},d_{n+1}^{(2)},f_4^{(1)}: \qquad\qquad \qquad\qquad\qquad\qquad\qquad\qquad\qquad\qquad\qquad
\qquad\qquad\qquad\qquad\nonumber\\
&\qquad\quad\rho^0_{ij}=c_i\,\overline{c}_i \quad \mbox{and} \quad \rho^0_{ji}=c_j\,\overline{c}_j(q+q^{-1})^2 \qquad \mbox{if} \quad i,j \ \mbox{are doubly linked with $i$ the longest root}\ ;\qquad\qquad\qquad\quad\ \nonumber\\
&\qquad\quad\rho^0_{ij}=c_i\,\overline{c}_i \quad \mbox{and}\quad \rho^0_{ji}=c_j\,\overline{c}_j\ 
\qquad\qquad \qquad \mbox{if} \quad i,j \ \mbox{are simply linked}\ . \quad \ \ \qquad\qquad\qquad\qquad  \qquad \qquad  \qquad \qquad\nonumber\\
&\qquad\mbox{For}\qquad {\widehat{g}}=g_2^{(1)}, d_{4}^{(3)}: \nonumber\\
&\qquad\quad \rho^0_{ij}=c_i\,\overline{c}_i, \quad \rho^0_{ji}=c_j\,\overline{c}_j (q^4+2\,q^2+4+2\,q^{-2}+q^{-4})\quad \mbox{and}\quad \rho^1_{ji}=-c^2_j\,\overline{c}^2_j\,(q^4+1+q^{-4})^2\,
\nonumber\\
&\qquad\qquad\mbox{if} \quad i,j \ \mbox{are triply linked with $i$ the longest root}\ ;\qquad\qquad\qquad\qquad\qquad\qquad\qquad\qquad\qquad\qquad\nonumber\\
&\qquad\quad\rho^0_{ij}=c_i\,\overline{c}_i \quad \mbox{and}\quad  \rho^0_{ji}=c_j\,\overline{c}_j
\qquad\mbox{if} \quad i,j \ \mbox{are simply linked}\ . \qquad\qquad\qquad\qquad\qquad\qquad\qquad \qquad\qquad \qquad \qquad\qquad \qquad  \qquad \qquad\nonumber\\
&\qquad\mbox{For}\qquad {\widehat{g}}=a_2^{(2)}:\nonumber\\
&\qquad\quad  \rho^0_{ij}=c_i\,\overline{c}_i, \quad  
\rho^0_{ji}=c_j\,\overline{c}_j(q+q^{-1})^2(q^4+3+q^{-4})\quad \mbox{and} \quad \rho^1_{ji}=-c^2_j\,\overline{c}^2_j(q+q^{-1})^4(q^2+q^{-2})^4\ \qquad\nonumber\\
&\qquad\qquad\qquad \mbox{with $i$ the longest root}\ . \qquad\qquad\qquad\qquad\qquad\qquad\qquad\qquad\qquad
\nonumber
\end{align}
\end{proof}

\begin{rem}
All the structure constants are invariant by the change $q\to q^{-1}$, which yields to the obvious  realization  $\textsf{A}_i=c_i\,e_iq_i^{-\frac{h_i}{2}} +\overline{c}_i\,f_iq_i^{-\frac{h_i}{2}} + w_i q_i^{-h_i}$.
\end{rem}

Quantum affine algebras $U_q(\widehat{g})$ are known to be Hopf algebras, thanks to the existence of a coproduct, counit and antipode actions. For generalized $q-$Onsager algebras $O_q(\widehat{g})$, a coaction map \cite{Chari} can be introduced:
\begin{prop} 
Let $c_i,{\overline c}_i\in{\mathbb C}$. The generalized q-Onsager algebra $O_q(\widehat{g})$ is a left $U_{q}(\widehat{g})-$comodule algebra with coaction map $\delta: O_q(\widehat{g})\rightarrow U_q(\widehat{g})\otimes O_q(\widehat{g})$
such that
\beqa
\delta(\textsf{A}_i)=  (c_i\,e_iq_i^{\frac{h_i}{2}} +\overline{c}_i\,f_iq_i^{\frac{h_i}{2}})\otimes I\!\!I + q_i^{h_i} \otimes \textsf{A}_i \ .\label{coactionmap}
\eeqa
\end{prop}
\begin{proof}
The verification of the comodule algebra axioms (see \cite{Chari}) is immediate using (\ref{coprod}). We have also to show that $\delta(\textsf{A}_i)$ statisfy (\ref{genqOns}). Assume ${\textsf A_i}$ satisfy (\ref{genqOns}). Plugging (\ref{coactionmap}) in (\ref{genqOns}), expanding and using the commutation relations of $U_q(\widehat{g})$ given in Definition 2.2, the claim follows. 
\end{proof}

\begin{rem}
If one embeds $O_q({\widehat g})$ into $U_q({\widehat g})$ according to Prop. 2.1, the coaction $\delta$ is identified with the comultiplication $\Delta$ of $U_q({\widehat g})$. 
\end{rem}

\vspace{2mm}

\section{Boundary affine Toda field theories revisited}
Among integrable quantum field theories, the sine-Gordon model is known to enjoy a hidden non-Abelian $U_q(\widehat{sl_2})$ symmetry, a fact that relies on the existence of {\it non-local} conserved charges generating the algebra \cite{BeLe}. Restricted to the half-line and perturbed at the boundary by certain local vertex operators, the boundary sine-Gordon model remains integrable \cite{GZ}. Corresponding scattering amplitudes of the fundamental solitons and breathers reflecting on the boundary have been derived either solving directly the so-called boundary Yang-Baxter equation (i.e. the reflection equation) \cite{GZ,Gh}, or using the existence of non-local conserved charges \cite{MN98,DM} that generate a remnant of the bulk $U_q(\widehat{sl_2})$ quantum group symmetry. However, the explicit defining relations of this remnant hidden non-Abelian symmetry algebra were only identified later on: for both integrable fixed or dynamical boundary conditions, the symmetry algebra is the $q-$Onsager algebra\,\footnote{For the XXZ open spin chain with generic integrable boundary conditions, the symmetry is associated with an (Abelian) $q-$Onsager's subalgebra. But in the thermodynamic limit, it is possible to show that the Hamiltonian becomes invariant under the action of the elements of the $q-$Onsager algebra \cite{BasS2}.} (\ref{Talg}) \cite{B1,B2}. In particular, in agreement with previous results fixed integrable boundary conditions are not restricted by the algebraic structure whereas dynamical ones \cite{BasDel,BasK3} are associated with boundary operators acting on finite or infinite dimensional representations of the $q-$Onsager algebra.\vspace{1mm}

Affine Toda field theories are natural generalizations of the sine-Gordon field theory, each being associated with an affine Lie algebra $\widehat{g}$. Similarly to the sine-Gordon case, in the bulk they enjoy a $U_q(\widehat{g})$ quantum group symmetry which determines completely all scattering amplitudes. Restricted on the half-line, 
two types of boundary conditions may be added that preserve integrability: either soliton non-preserving - the most studied case\footnote{Among the known non-perturbative results in boundary affine Toda field theories, scattering amplitudes (for $a_n^{(1)},n>1$) have been considered in details in \cite{Gand,DelGan}, and mass-parameter as well as vacuum expectation values of local fields have been proposed in \cite{Fateev}. See also related results and non-perturbative checks in \cite{Ahn}.} since \cite{Fring,Cor} - or soliton preserving \cite{Skly88,DelSol} boundary conditions. 
In the following, we focus on the first family of integrable models which Euclidean action\footnote{According to a recent paper \cite{Doik} (see also \cite{Doikou1}), a Hamiltonian has been proposed for soliton preserving boundary conditions.} reads \cite{Fring,Cor}:
\begin{align}
S=\frac{1}{4\pi}\int_{x<0} d^2z \Big( \partial\phi\overline{\partial}\phi + \frac{\lambda}{2\pi} \sum_{j=0}^{n} n_j\exp\big(-i\hat{\beta}\frac{1}{|\alpha_j|^2}\alpha_j \cdot \phi\big)\Big) + \frac{\lambda_b}{2\pi} \int dt \sum_{j=0}^{n} \epsilon_j\exp\big(-i\frac{\hat{\beta}}{2}\alpha_j \cdot \phi(0,t)\big) \ ,\label{Toda}
\end{align}
where $\phi(x,t)$ is an $n-$component bosonic field in two dimensions, $\{\alpha_j\}$ and $n_j$ are the simple roots and Kac labels, respectively, of ${\widehat g}$, $\lambda,\lambda_b$ are related with the mass scale, $\widehat{\beta}$ is the coupling constant and $\{\epsilon_j\}$ are the boundary parameters or operators.
This action remains however integrable for certain scalar boundary conditions $\epsilon_j$ that have been identified either at the classical \cite{Cor} or quantum \cite{Penati} level based on the existence of {\it local} higher spin conserved charges\,\footnote{At classical level, an extended Lax pair formalism has also been proposed \cite{Cor}. Given few assumptions, it gives further support for the boundary conditions previously derived.}. For the simply laced cases $\widehat{g}=a_n^{(1)},d_n^{(1)}$, {\it non-local} conserved charges that generate a (coideal) subalgebra of $U_q(\widehat{g})$ have also been derived \cite{DM,DelG}. They read:
\begin{align}
{\hat Q}_j= Q_j + \overline{Q}_j +\widehat{\epsilon}_j q^{T_j} \ ,\qquad j=0,1,,...,n\ \qquad \mbox{with} \qquad \widehat{\epsilon}_j=\frac{\lambda_b}{2\pi c}\frac{\hat{\beta}^2}{1-\hat{\beta}^2}\epsilon_j \label{Qhat}
\end{align}
where $c=\sqrt{\lambda(\hat{\beta}^2/(2-\hat{\beta}^2))^2(q^2-1)/2i\pi}$, the charges $Q_j,\overline{Q}_j$ are realized in terms of vertex operators of holomorphic/antiholomorphic fields and $T_j$ has a form analog to the bulk topological charge but restricted to the half-line. For more details, explicit expressions can be found in \cite{DM}. Generalizations of the expressions (\ref{Qhat}) to the non-simply laced cases are straightforward. Although in \cite{DM} only scalar boundary conditions were considered, calculations leading to (\ref{Qhat}) also hold assuming instead boundary operators $\epsilon_j$ provided
\begin{align}
[x,\widehat{\epsilon}_j]=0\qquad \ \forall x\in\{Q_j,\overline{Q}_j,T_j\}\ .
\end{align}

Despite of the results in \cite{DM,DelG} that provide a powerful tool to derive efficiently all scattering amplitudes of the solitons reflecting on the boundary\,\footnote{Deriving all scattering amplitudes solely using the reflection equation - as done in \cite{Gand} for the case $a_n^{(1)}$ - is more difficult.}, the explicit defining relations of the $U_q(\widehat{g})$'s coideal subalgebra generalizing (\ref{Talg}) to higher rank $\widehat{g}$ are still unknown up to now. Beyond the interest of having a proper mathematical frame, this problem is relevant in the study of (\ref{Toda}) as admissible fixed or dynamical boundary conditions and boundary states should be classified according to the representation theory of the algebra generated by (\ref{Qhat}).\vspace{1mm}

To identify the underlying non-Abelian hidden symmetry of (\ref{Toda}) and classify corresponding boundary conditions, in both situations (fixed or dynamical boundary conditions) it is then important to recall that the existence of non-local conserved charges of the form (\ref{Qhat}) with $q\rightarrow q_j$ for non-simply laced cases essentially relies on the structure of the boundary terms appearing in (\ref{Toda}). Given such boundary terms and having derived the non-local conserved charges \cite{MN98,DM,DelG}, the integrability condition for the model (\ref{Toda}) requires that all non-local charges close among a finite number of algebraic relations, yet to be identified. The answer to this problem - finding these algebraic relations - actually follows from the results of the previous Section. Indeed, presented in terms of $U_q(\widehat{g})$ generators and up to an overall scalar factor the non-local conserved charges (\ref{Qhat}) with $q\rightarrow q_j$ for non-simply laced cases  turn out to be exactly of the form (\ref{realg}). Then, two situations can be considered:

\subsection{Fixed boundary conditions}

Assume $\widehat{\epsilon}_j$ (or equivalently $\epsilon_j$) are scalars. According to Proposition 2.1, given any simply or non-simply laced affine Lie algebra the corresponding non-local conserved charges (\ref{Qhat}) with $q\rightarrow q_j$ close over the relations (\ref{genqOns}) provided the boundary conditions are constrained by the relations below (\ref{realg}) setting $c_j=\overline{c}_j\equiv 1$, $\widehat{\epsilon}_j\equiv w_j$. Appart from the simple solutions $w_j\equiv 0$ \ $\forall j$, solving all constraints case by case yields to the following families of admissible integrable boundary conditions:
\begin{align}
\mbox{For}\qquad {\widehat{g}}=a_n^{(1)} (n>1), d_n^{(1)}, e_6^{(1)},e_7^{(1)},e_8^{(1)}: 
\qquad \widehat{\epsilon}_j=\pm\frac{i}{q^{1/2}-q^{-1/2}} \quad \forall j\ ;\qquad\qquad\qquad\qquad\qquad\qquad\ \ \nonumber
\end{align}
\begin{align}
\mbox{For}\qquad {\widehat{g}}=b_n^{(1)}: \qquad \widehat{\epsilon}_j=\pm\frac{i}{q-q^{-1}} \quad \mbox{for}\  j\in\{0,1,...,n-1\}\ \ , \quad \widehat{\epsilon}_n \ \mbox{arbitrary}\ ;\qquad\qquad\qquad\qquad\qquad\qquad\ \ \nonumber
\end{align}
\begin{align}
\mbox{For}\qquad {\widehat{g}}=a_{2n-1}^{(2)}: \quad \Big\{\begin{array}{c}\mbox{either} \qquad \widehat{\epsilon}_j=\pm\frac{i}{q_j^{1/2}-q_j^{-1/2}} \quad \mbox{for}\  j\in\{0,1,...,n\}\  \qquad\qquad\qquad\qquad\qquad\qquad\quad\quad\nonumber\\
\quad \ \ \ \mbox{or}\qquad \  \widehat{\epsilon}_j=0 \qquad\qquad\qquad\mbox{for}\  j\in\{0,1,...,n-1\}\ \ , \qquad  \ \widehat{\epsilon}_n \ \mbox{arbitrary}\ ;\qquad\qquad\end{array}
\end{align}
\begin{align}
\mbox{For}\qquad {\widehat{g}}=c_n^{(1)}: \quad \Big\{\begin{array}{c}\mbox{either} \qquad \widehat{\epsilon}_j=\pm\frac{i}{q_j^{1/2}-q_j^{-1/2}} \quad \mbox{for}\  j\in\{0,...,n\}\ \ , \qquad\qquad\qquad\qquad\qquad\qquad\quad\quad \nonumber\\
\quad \mbox{or}\qquad \  \widehat{\epsilon}_j=0 \qquad\qquad\qquad\mbox{for}\  j\in\{1,...,n-1\}\ \ , \qquad  \ \widehat{\epsilon}_0, \widehat{\epsilon}_n \ \mbox{arbitrary}\ ;\qquad\qquad\end{array}
\end{align}
\begin{align}
\mbox{For}\qquad {\widehat{g}}=d_{n+1}^{(2)}: \qquad \widehat{\epsilon}_j=\pm\frac{i}{q_j-q_j^{-1}} \quad \mbox{for}\  j\in\{1,...,n-1\}\ \ , \quad \widehat{\epsilon}_0,\widehat{\epsilon}_n \ \mbox{arbitrary}\ ;\qquad\qquad\qquad\qquad\quad\ \  \nonumber
\end{align}

\begin{align}
\mbox{For}\qquad {\widehat{g}}=a_{2n}^{(2)} (n>2): \quad \Big\{\begin{array}{c}\mbox{either} \qquad \widehat{\epsilon}_j=\pm\frac{i}{q_j^{1/2}-q_j^{-1/2}}\ \quad \mbox{for}\  j\in\{1,...,n\}\ \qquad\qquad  \ \widehat{\epsilon}_0\ \mbox{arbitrary} \qquad\qquad\qquad\nonumber\\
\ \mbox{or}\qquad \  \widehat{\epsilon}_j=0 \qquad\qquad\qquad\mbox{for}\  j\in\{0,...,n-1\}\ \ , \qquad  \ \widehat{\epsilon}_n\ \mbox{arbitrary}\ ;\qquad\qquad\end{array}
\end{align}
\begin{align}
\mbox{For}\qquad {\widehat{g}}=a_{2}^{(2)}: \quad \Big\{\begin{array}{c}\mbox{either} \qquad \widehat{\epsilon}_0=\pm\frac{i}{q^{2}-q^{-2}}\ ,\qquad  \ \widehat{\epsilon}_1\ \mbox{arbitrary} \qquad\qquad\qquad\nonumber\\
\ \mbox{or}\qquad \  \widehat{\epsilon}_1=0\ , \qquad  \ \widehat{\epsilon}_0\ \mbox{arbitrary}\ ;\qquad\qquad\qquad\qquad\end{array}\qquad\qquad\qquad\qquad\qquad
\end{align}
\begin{align}
\mbox{For}\qquad {\widehat{g}}=a_{4}^{(2)}: \quad \Big\{\begin{array}{c}\mbox{either} \qquad \widehat{\epsilon}_j=\pm\frac{i}{q_j^{1/2}-q_j^{-1/2}}\ \quad \mbox{for}\  j=1,2\ , \widehat{\epsilon}_0 \ \mbox{arbitrary}\  \qquad\qquad\qquad\qquad\qquad\qquad\quad\quad \nonumber\\
\ \mbox{or}\qquad \  \widehat{\epsilon}_2=\pm\frac{i}{q^{2}-q^{-2}}\ , \widehat{\epsilon}_0=0\ , \ \widehat{\epsilon}_1\ \mbox{arbitrary} \qquad\qquad\qquad\qquad\qquad\qquad\quad\qquad \ \ \ \nonumber\\ 
\ \mbox{or}\qquad \ \widehat{\epsilon}_0=\widehat{\epsilon}_1=0\ ,\ \widehat{\epsilon}_2\ \mbox{arbitrary}\ ;\qquad\qquad\qquad\qquad\qquad\qquad\qquad\qquad\qquad\qquad\end{array}
\end{align}
\begin{align}
\mbox{For}\qquad {\widehat{g}}=g_{2}^{(1)}: \quad \Big\{\begin{array}{c}\mbox{either} \qquad \widehat{\epsilon}_j=\pm\frac{i}{q_j^{1/2}-q_j^{-1/2}}\ \qquad\qquad\qquad\qquad\qquad\qquad\quad\quad \qquad\qquad\qquad\qquad\qquad\nonumber\\
\ \mbox{or}\qquad \  \widehat{\epsilon}_j=\pm\frac{i}{q_j^{1/2}-q_j^{-1/2}}\quad \mbox{for}\  j=0,1\ , \quad\widehat{\epsilon}_2= \pm\frac{i(q+q^{-1}-1)}{q^{1/2}-q^{-1/2}}\ ;\qquad\qquad \qquad\qquad\end{array}
\end{align}
\begin{align}
\mbox{For}\qquad {\widehat{g}}=d_{4}^{(3)} : \quad \ \widehat{\epsilon}_j=\pm\frac{i}{q_j^{1/2}-q_j^{-1/2}}\ \quad \mbox{for}\  j\in\{0,1,2\}\ \qquad\qquad\qquad\qquad\qquad\qquad\qquad\qquad\qquad\quad\quad \nonumber\ ;
\end{align}
\begin{align}
\mbox{For}\qquad {\widehat{g}}=f_{4}^{(1)} : \quad \Big\{\begin{array}{c}\mbox{either} \ \ \ \ \ \widehat{\epsilon}_j=\pm\frac{i}{q_j^{1/2}-q_j^{-1/2}}\ \quad \mbox{for}\  j\in\{0,...,4\}\ \qquad\qquad\qquad\qquad\qquad\qquad\qquad\qquad\qquad\quad\quad \nonumber\\
\ \ \ \mbox{or}\ \ \quad \widehat{\epsilon}_j=\pm\frac{i}{q_j-q_j^{-1}} \qquad\ \quad\mbox{for}\  j\in\{0,1,2\}\ , \quad  \ \widehat{\epsilon}_j= 0\quad \mbox{for}\ j\in\{3,4\}   \ ;\qquad\qquad\quad\qquad\qquad\end{array}
\end{align}
\begin{align}
\mbox{For}\qquad {\widehat{g}}=e_{6}^{(2)} : \quad \Big\{\begin{array}{c}\mbox{either} \ \ \ \ \ \widehat{\epsilon}_j=\pm\frac{i}{q_j^{1/2}-q_j^{-1/2}}\ \quad \mbox{for}\  j\in\{0,...,4\}\ \qquad\qquad\qquad\qquad\qquad\qquad\qquad\qquad\qquad\quad\quad \nonumber\\
\ \mbox{or}\ \ \quad \widehat{\epsilon}_j=0 \quad\qquad\qquad\quad\mbox{for}\  j\in\{0,1,2\}\ , \quad  \ \widehat{\epsilon}_j= \pm\frac{i}{q_j-q_j^{-1}}\quad \mbox{for}\  j\in\{3,4\}   \ .\qquad\qquad\end{array}
\end{align}
\vspace{2mm}
Note that for the cases $a_n^{(1)},d_n^{(1)}$, above results are in perfect agreement with \cite{DM,DelG}. In the classical limit $q\rightarrow 1$,  except for the exceptional cases $g_2^{(1)},d_4^{(3)}$ all above integrable boundary conditions agree with the results in \cite{Cor}. 

\subsection{Dynamical boundary conditions}
By analogy with \cite{BasK3}, instead of scalar boundary conditions an interesting problem is to consider additional operators $\widehat{\epsilon}_j$ located at the boundary, and interacting with the bulk fields according to (\ref{Toda}). As mentionned above, following the arguments of \cite{DM} non-local conserved charges of the form (\ref{Qhat}) with $q\rightarrow q_j$ can be constructed. These charges can be written:
\begin{align}
{\hat Q}_j= (Q_j + \overline{Q}_j)\otimes I\!\!I + q_j^{T_j} \otimes \widehat{\epsilon}_j \ ,\qquad j=0,1,,...,n\ \label{Qhat2}
\end{align}
where the first and second representation spaces are associated with the particle/boundary space of states, respectively. 
Integrability requires that the charges (\ref{Qhat2}) form an algebra, ensuring the existence of a factorized scattering theory and, in particular, of a solition reflection matrix commuting with (\ref{Qhat2}). We are then looking for a set of algebraic relations satisfied by the elements (\ref{Qhat2}).
To this end, it is crucial to notice the following: according to the defining relations of $U_q({\widehat g})$ and the term $(Q_j + \overline{Q}_j)\otimes I\!\!I$ in (\ref{Qhat2}) such non linear combinations of (\ref{Qhat2}) for different $j$ can only simplify if $q-$Serre relations are used. More precisely, a straightforward calculation shows that combinations of $(Q_j + \overline{Q}_j)\otimes I\!\!I$ for different $j$ only close on the algebraic relations (\ref{genqOns}) - a consequence of Proposition 2.1 for $w_i\equiv 0$ $\forall i$. So, if the conserved charges form an algebra, due to the term  $(Q_j + \overline{Q}_j)\otimes I\!\!I$ its defining relations are necessarely given by (\ref{genqOns}). Let us then see under which conditions on the boundary operators $\widehat{\epsilon}_j$ the whole combination (\ref{Qhat2}) could satisfy (\ref{genqOns}) setting ${\textsf A}_j\equiv {\hat Q}_j$. Plugging (\ref{Qhat2}) in (\ref{genqOns}) and expanding one finds that (\ref{Qhat2}) satisfy the algebraic relations (\ref{genqOns}) if and only if the terms $\widehat{\epsilon}_j$ also satisfy (\ref{genqOns}). 
Note that these calculations are analogous to the ones of Proposition 2.2, which explains the form of the coaction map as defined in (\ref{coactionmap}). Under these conditions, it follows that ${\hat Q}_j$ generate the $q-$Onsager algebra. For $g=sl_2$, a simple realization has been proposed in \cite{BasK3}. For higher rank cases, an interesting problem would be to construct realizations in terms of $q-$deformed oscillators, generalizing the results of the massless case (see eq. (1.17) in \cite{BHK}). In any case, given the family of boundary integrable affine Toda field theories (\ref{Toda}) all admissible dynamical boundary conditions $\widehat{\epsilon}_j$ are required to satisfy (\ref{genqOns}). 

\section{Discussion}
In this letter, a new family of quantum algebras that we call the {\it generalized $q-$Onsager algebras} $O_q(\widehat{g})$ associated with the affine Lie algebras $\widehat{g}$ has been introduced and studied. Some properties and the explicit relationship with coideal subalgebras of $U_q(\widehat{g})$ have been clarified, and simple consequences for quantum integrable systems - namely boundary affine Toda field theories with soliton non-preserving boundary conditions - have been explored. Clearly, extending all known results of the $\widehat{sl_2}-$case (\ref{Talg}) to the whole family $O_q(\widehat{g})$ is rather interesting from different points of view.\vspace{1mm}

From the mathematical side, it is now well understood thanks to Terwilliger {\it et al.}'s works (see some references below) that (\ref{Talg}) provides an algebraic framework to classify all othogonal polynomials of the Askey scheme. Weither the generalized $q-$Onsager algebras $O_q(\widehat{g})$ provide an algebraic framework for multivariable orthogonal polynomials - known or new - is an interesting problem. Another interesting problem is to construct new current algebras associated with $O_q(\widehat{g})$ by analogy with
\cite{BasS} and establish the isomorphism bewteen $O_q(\widehat{g})$ and the family of reflection equation algebra associated with $q-$twisted Yangians \cite{Mol} for $R-$matrices associated with higher rank quantum affine Lie algebras. \vspace{1mm}

From the physics side - beyond the explicit construction of boundary reflection matrices for boundary affine Toda field theories (see \cite{DM,DelG} for the simply laced cases) - generalized $q-$Onsager algebras should provide a powerful tool in order to study quantum integrable systems with extended symmetries. In this direction, irreducible representations of $O_q(\widehat{g})$ will find applications to the spectrum of boundary states in boundary integrable quantum field theories.
Also, studying the explicit construction of a hierarchy of commuting quantities that generalizes the Dolan-Grady hierarchy \cite{DG} or its $q-$deformed analogue \cite{B1} will find applications in studying the spectrum and eigenstates in related spin chains. The results in \cite{Doikou} might be a good starting point.

Some of these problems will be considered elsewhere.\vspace{2mm}

\noindent{\bf Acknowledgements:} S.B thanks LMPT for hospitality where part of this work has been done, INFN iniziativa specifica FI11 for financial support and Italian Ministry of Education, University and Research grant PRIN-2007JHLPEZ.
\vspace{1cm}

{\bf Appendix A. Dynkin diagrams for affine Lie algebras}\\
The upper (resp. lower) indices denote the number (resp. value of ($d_i$,$n_i$)) associated with each node.  The explicit values of the coefficient of the extended Cartan matrix $a_{ij}$ for each affine Lie algebra \cite{Kac} can be found using $a_{ii}=2$ and the rules:
\\

\begin{picture}(40,10) \thicklines 
	\multiput(0,0)(28,0){2}{\circle{10}}
	\put(0,12){\makebox(0,0){{\footnotesize i}}} 
	\put(28,12){\makebox(0,0){{\footnotesize j}}} 
\end{picture}
$a_{ij}=a_{ji}=0$,\qquad
\begin{picture}(40,10) \thicklines 
\multiput(0,0)(28,0){2}{\circle{10}} 
\put(4,3){\line(1,0){20}}\put(4,-3){\line(1,0){20}} 
\put(0,12){\makebox(0,0){{\footnotesize i}}} 
\put(28,12){\makebox(0,0){{\footnotesize j}}} 
\end{picture}
$a_{ij}=a_{ji}=-2$,\qquad
\begin{picture}(40,10) \thicklines 
\multiput(0,0)(28,0){2}{\circle{10}} 
\put(5,0){\line(1,0){18}} 
\put(0,12){\makebox(0,0){{\footnotesize i}}}
\put(28,12){\makebox(0,0){{\footnotesize j}}} 
\end{picture}
$a_{ij}=a_{ji}=-1$,\qquad
\\

\begin{picture}(40,10) \thicklines 
\multiput(0,0)(28,0){2}{\circle{10}} 
\put(4,3){\line(1,0){20}}\put(4,-3){\line(1,0){20}} 
\put(18,0){\line(-1,1){10}}\put(18,0){\line(-1,-1){10}}
\put(0,12){\makebox(0,0){{\footnotesize i}}} 
\put(28,12){\makebox(0,0){{\footnotesize j}}} 
\end{picture}
$a_{ij}=-1$\quad $a_{ji}=-2$,\qquad
\begin{picture}(40,10) \thicklines 
\multiput(0,0)(28,0){2}{\circle{10}} 
\put(4,3){\line(1,0){20}}\put(4,-3){\line(1,0){20}} 
\put(18,0){\line(-1,1){10}}\put(18,0){\line(-1,-1){10}}
\put(5,0){\line(1,0){18}} 
\put(0,12){\makebox(0,0){{\footnotesize i}}} 
\put(28,12){\makebox(0,0){{\footnotesize j}}} 
\end{picture}
$a_{ij}=-1$\quad  $a_{ji}=-3$,\qquad
\begin{picture}(40,10) \thicklines 
	     \multiput(0,0)(28,0){2}{\circle{10}}
		\put(2,-5){\line(1,0){24}} 
		\put(4,-3){\line(1,0){20}}
		\put(4,3){\line(1,0){20}} 
		\put(2,5){\line(1,0){24}} 
		\put(18,0){\line(-1,-1){10}}
		\put(18,0){\line(-1,1){10}}
		\put(0,12){\makebox(0,0){{\footnotesize i}}} 
\put(28,12){\makebox(0,0){{\footnotesize j}}} 
\end{picture}
$a_{ij}=-1$\quad  $a_{ji}=-4$.
\vspace{0.2cm}
\\

- Simply laced Dynkin diagrams (all $d_i=1$ and the lower indices correspond to $(n_i)$):

$a^{(1)}_1$  \qquad  \qquad
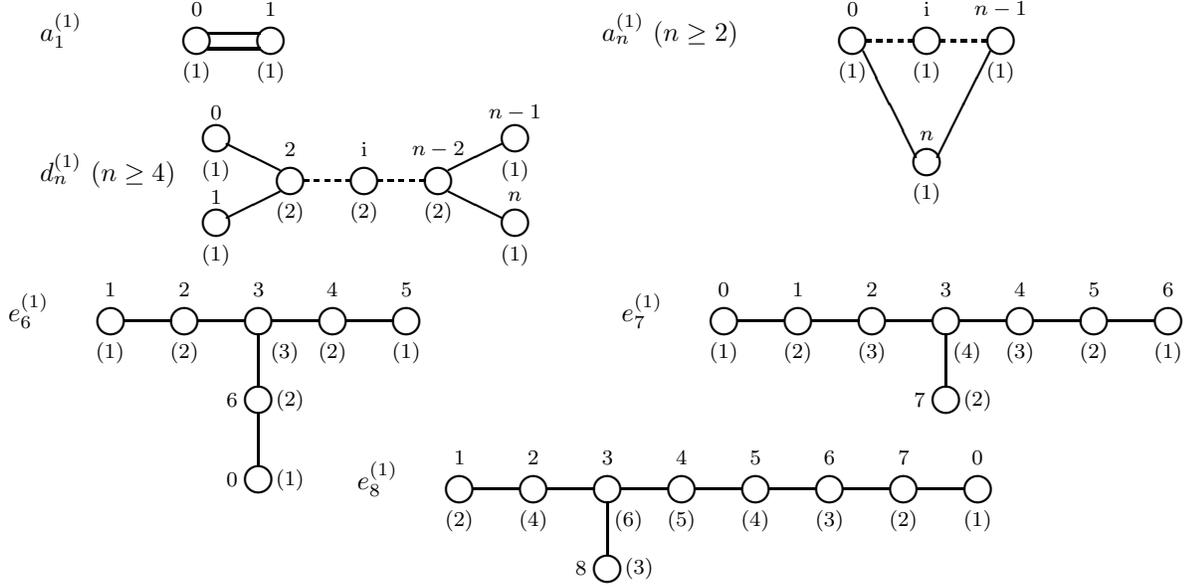
\begin{picture}(70,30) \thicklines 
\multiput(0,0)(28,0){2}{\circle{10}} 
\put(4,3){\line(1,0){20}}\put(4,-3){\line(1,0){20}} 
\put(0,12){\makebox(0,0){{\footnotesize 0}}} 
\put(28,12){\makebox(0,0){{\footnotesize 1}}} 
\put(0,-12){\makebox(0,0){{\footnotesize (1)}}}
\put(28,-12){\makebox(0,0){{\footnotesize (1)}}} 
\end{picture}
\qquad \qquad \qquad \qquad
$a^{(1)}_n \ (n\geq2)$ \qquad  \qquad	\begin{picture}(70,30) \thicklines 
		\put(28,-46){\circle{10}}
		\put(4,-4){\line(1,-2){20}}\put(52,-4){\line(-1,-2){20}}
		\multiput(0,0)(28,0){3}{\circle{10}} 
		\multiput(5,0)(4,0){5}{\line(1,0){2}}
		\multiput(33,0)(4,0){5}{\line(1,0){2}}
		\put(0,12){\makebox(0,0){{\footnotesize 0}}} 
		\put(28,12){\makebox(0,0){{\footnotesize i}}} 
		\put(56,12){\makebox(0,0){{\footnotesize $n-1$}}} 
		\put(28,-36){\makebox(0,0){{\footnotesize $n$}}} 
		\put(0,-12){\makebox(0,0){{\footnotesize (1)}}} 
		\put(28,-12){\makebox(0,0){{\footnotesize (1)}}} 
		\put(56,-12){\makebox(0,0){{\footnotesize (1)}}} 
		\put(28,-58){\makebox(0,0){{\footnotesize (1)}}} 
	\end{picture}
\\

$d^{(1)}_n  \ (n\geq4)$ \qquad\qquad
	\begin{picture}(90,40) \thicklines 
	         \put(-28,16){\circle{10}}\put(-28,-16){\circle{10}} 
		\put(-4,-4){\line(-2,-1){20}}\put(-4,4){\line(-2,1){20}}
		\multiput(0,0)(28,0){3}{\circle{10}} 
		\multiput(5,0)(4,0){5}{\line(1,0){2}}
		\multiput(33,0)(4,0){5}{\line(1,0){2}}
		\put(85,16){\circle{10}}\put(85,-16){\circle{10}} 
		\put(60,4){\line(2,1){20}}\put(60,-4){\line(2,-1){20}}
		\put(0,12){\makebox(0,0){{\footnotesize 2}}} 
		\put(28,12){\makebox(0,0){{\footnotesize i}}} 
		\put(56,12){\makebox(0,0){{\footnotesize $n-2$}}} 
		\put(0,-12){\makebox(0,0){{\footnotesize (2)}}} 
		\put(28,-12){\makebox(0,0){{\footnotesize (2)}}} 
		\put(56,-12){\makebox(0,0){{\footnotesize (2)}}} 
		\put(85,26){\makebox(0,0){{\footnotesize  $n-1$}}} 
		\put(85,-6){\makebox(0,0){{\footnotesize  $n$}}} 
		\put(85,4){\makebox(0,0){{\footnotesize (1)}}} 
		\put(85,-28){\makebox(0,0){{\footnotesize (1)}}} 
		\put(-28,26){\makebox(0,0){{\footnotesize  0}}} 
		\put(-28,-6){\makebox(0,0){{\footnotesize  1}}} 
		\put(-28,4){\makebox(0,0){{\footnotesize (1)}}} 
		\put(-28,-28){\makebox(0,0){{\footnotesize (1)}}} 
	\end{picture}
\\	
\\
$e^{(1)}_6$  \qquad \begin{picture}(130,40) \thicklines 
		\multiput(0,0)(28,0){5}{\circle{10}} 
		\multiput(5,0)(28,0){4}{\line(1,0){18}} 
		\put(56,-5){\line(0,-1){20}} 
		\put(56,-30){\circle{10}}
		\put(56,-35){\line(0,-1){20}} 
		\put(56,-60){\circle{10}}
		\put(0,12){\makebox(0,0){{\footnotesize 1}}} 
		\put(28,12){\makebox(0,0){{\footnotesize 2}}} 
		\put(56,12){\makebox(0,0){{\footnotesize 3}}} 
		\put(84,12){\makebox(0,0){{\footnotesize 4}}} 
		\put(112,12){\makebox(0,0){{\footnotesize 5}}} 
		\put(46,-30){\makebox(0,0){{\footnotesize 6}}} 
		\put(46,-60){\makebox(0,0){{\footnotesize 0}}} 
		\put(68,-60){\makebox(0,0){{\footnotesize (1)}}} 
		\put(0,-12){\makebox(0,0){{\footnotesize (1)}}} 
		\put(28,-12){\makebox(0,0){{\footnotesize (2)}}} 
		\put(66,-12){\makebox(0,0){{\footnotesize (3)}}} 
		\put(84,-12){\makebox(0,0){{\footnotesize (2)}}} 
		\put(112,-12){\makebox(0,0){{\footnotesize (1)}}} 
		\put(68,-30){\makebox(0,0){{\footnotesize (2)}}} 
	\end{picture}
\qquad \qquad \qquad
$e^{(1)}_7$ \qquad	\begin{picture}(140,40) \thicklines 
		\multiput(0,0)(28,0){7}{\circle{10}} 
		\multiput(5,0)(28,0){6}{\line(1,0){18}} 
		\put(84,-5){\line(0,-1){20}} 
		\put(84,-30){\circle{10}}
		\put(0,-12){\makebox(0,0){{\footnotesize (1)}}} 
		\put(28,-12){\makebox(0,0){{\footnotesize (2)}}} 
		\put(56,-12){\makebox(0,0){{\footnotesize (3)}}} 
		\put(92,-12){\makebox(0,0){{\footnotesize (4)}}} 
		\put(112,-12){\makebox(0,0){{\footnotesize (3)}}} 
		\put(140,-12){\makebox(0,0){{\footnotesize (2)}}} 
		\put(168,-12){\makebox(0,0){{\footnotesize (1)}}} 
		\put(96,-30){\makebox(0,0){{\footnotesize (2)}}} 
		\put(0,12){\makebox(0,0){{\footnotesize 0}}} 
		\put(28,12){\makebox(0,0){{\footnotesize 1}}} 
		\put(56,12){\makebox(0,0){{\footnotesize 2}}} 
		\put(84,12){\makebox(0,0){{\footnotesize 3}}} 
		\put(112,12){\makebox(0,0){{\footnotesize 4}}} 
		\put(140,12){\makebox(0,0){{\footnotesize 5}}} 
		\put(168,12){\makebox(0,0){{\footnotesize 6}}} 
		\put(74,-30){\makebox(0,0){{\footnotesize 7}}} 
	\end{picture}

\qquad \qquad \qquad \qquad \qquad \qquad $e^{(1)}_8$ \qquad	\begin{picture}(130,60) \thicklines 
		\multiput(0,0)(28,0){8}{\circle{10}} 
		\multiput(5,0)(28,0){7}{\line(1,0){18}} 
		\put(56,-5){\line(0,-1){20}} 
		\put(56,-30){\circle{10}}
		\put(0,-12){\makebox(0,0){{\footnotesize (2)}}} 
		\put(28,-12){\makebox(0,0){{\footnotesize (4)}}} 
		\put(64,-12){\makebox(0,0){{\footnotesize (6)}}} 
		\put(84,-12){\makebox(0,0){{\footnotesize (5)}}} 
		\put(112,-12){\makebox(0,0){{\footnotesize (4)}}} 
		\put(140,-12){\makebox(0,0){{\footnotesize (3)}}} 
		\put(168,-12){\makebox(0,0){{\footnotesize (2)}}} 
		\put(68,-30){\makebox(0,0){{\footnotesize (3)}}}
		\put(196,-12){\makebox(0,0){{\footnotesize (1)}}} 
		\put(0,12){\makebox(0,0){{\footnotesize 1}}} 
		\put(28,12){\makebox(0,0){{\footnotesize 2}}} 
		\put(56,12){\makebox(0,0){{\footnotesize 3}}} 
		\put(84,12){\makebox(0,0){{\footnotesize 4}}} 
		\put(112,12){\makebox(0,0){{\footnotesize 5}}} 
		\put(140,12){\makebox(0,0){{\footnotesize 6}}} 
		\put(168,12){\makebox(0,0){{\footnotesize 7}}} 
		\put(196,12){\makebox(0,0){{\footnotesize 0}}} 
		\put(46,-30){\makebox(0,0){{\footnotesize 8}}}  
	\end{picture}
	\\
\\
\\
\\
- Non-simply laced Dynkin diagrams (the lower numbers correspond to $(d_i,n_i)$):
\\
	$b^{(1)}_n \ (n\geq3)$\qquad \qquad
		\begin{picture}(90,30) \thicklines 
		\put(-28,16){\circle{10}}\put(-28,-16){\circle{10}} 
		\put(-4,-4){\line(-2,-1){20}}\put(-4,4){\line(-2,1){20}}
		\multiput(0,0)(28,0){4}{\circle{10}} 
		\put(5,0){\line(1,0){18}} 
		\multiput(33,0)(4,0){5}{\line(1,0){2}}
		\multiput(33,0)(4,0){5}{\line(1,0){2}}
		\put(60,3){\line(1,0){20}}\put(60,-3){\line(1,0){20}} 
		\put(75,0){\line(-1,1){10}}\put(75,0){\line(-1,-1){10}}
		\put(-28,26){\makebox(0,0){{\footnotesize  0}}} 
		\put(-28,-6){\makebox(0,0){{\footnotesize  1}}} 
		\put(-28,4){\makebox(0,0){{\footnotesize (2,1)}}} 
		\put(-28,-28){\makebox(0,0){{\footnotesize (2,1)}}} 
		\put(0,12){\makebox(0,0){{\footnotesize 2}}} 
		\put(28,12){\makebox(0,0){{\footnotesize 3}}} 
		\put(56,12){\makebox(0,0){{\footnotesize $n$ -1}}} 
		\put(84,12){\makebox(0,0){{\footnotesize $n$}}}
		\put(0,-12){\makebox(0,0){{\footnotesize (2,2)}}} 
		\put(28,-12){\makebox(0,0){{\footnotesize (2,2)}}} 
		\put(56,-12){\makebox(0,0){{\footnotesize (2,2)}}} 
		\put(84,-12){\makebox(0,0){{\footnotesize (1,2)}}}  
	\end{picture}
\qquad \qquad
	$a^{(2)}_{2n-1} \ (n\geq3)$ \qquad\qquad
	\begin{picture}(90,40) \thicklines 
	         \put(-28,16){\circle{10}}\put(-28,-16){\circle{10}} 
		\put(-4,-4){\line(-2,-1){20}}\put(-4,4){\line(-2,1){20}}
		\multiput(0,0)(28,0){5}{\circle{10}} 
		\put(5,0){\line(1,0){18}}
		\multiput(33,0)(4,0){5}{\line(1,0){2}}
		\put(60,0){\line(1,0){18}} 
		\put(88,3){\line(1,0){20}}\put(88,-3){\line(1,0){20}} 
		\put(93,0){\line(1,1){10}}\put(93,0){\line(1,-1){10}}
		\put(-28,26){\makebox(0,0){{\footnotesize  0}}} 
		\put(-28,-6){\makebox(0,0){{\footnotesize  1}}} 
		\put(-28,4){\makebox(0,0){{\footnotesize (1,1)}}} 
		\put(-28,-28){\makebox(0,0){{\footnotesize (1,1)}}} 
		\put(0,12){\makebox(0,0){{\footnotesize 2}}} 
		\put(28,12){\makebox(0,0){{\footnotesize 3}}} 
		\put(56,12){\makebox(0,0){{\footnotesize $n$ -2}}} 
		\put(84,12){\makebox(0,0){{\footnotesize $n-1$}}}
		\put(112,12){\makebox(0,0){{\footnotesize $n$}}}
		\put(0,-12){\makebox(0,0){{\footnotesize (1,2)}}} 
		\put(28,-12){\makebox(0,0){{\footnotesize (1,2)}}} 
		\put(56,-12){\makebox(0,0){{\footnotesize (1,2)}}} 
		\put(84,-12){\makebox(0,0){{\footnotesize (1,2)}}}  
		\put(112,-12){\makebox(0,0){{\footnotesize (2,1)}}} 
	\end{picture}
\\
\\
$c^{(1)}_n \ (n\geq2)$ \qquad \qquad
	\begin{picture}(100,40) \thicklines 
	\put(-28,0){\circle{10}}
	\put(-24,3){\line(1,0){20}}\put(-24,-3){\line(1,0){20}} 
		\put(-10,0){\line(-1,1){10}}\put(-10,0){\line(-1,-1){10}}
		\multiput(0,0)(28,0){5}{\circle{10}} 
		\put(5,0){\line(1,0){18}} 
		\multiput(33,0)(4,0){5}{\line(1,0){2}}
		\put(60,0){\line(1,0){18}} 
		\put(88,3){\line(1,0){20}}\put(88,-3){\line(1,0){20}} 
		\put(93,0){\line(1,1){10}}\put(93,0){\line(1,-1){10}}
		\put(-28,12){\makebox(0,0){{\footnotesize 0}}} 
		\put(-28,-12){\makebox(0,0){{\footnotesize (2,1)}}} 
		\put(0,12){\makebox(0,0){{\footnotesize 1}}} 
		\put(0,-12){\makebox(0,0){{\footnotesize (1,2)}}} 
		\put(28,12){\makebox(0,0){{\footnotesize 2}}} 
		\put(28,-12){\makebox(0,0){{\footnotesize (1,2)}}} 
		\put(56,12){\makebox(0,0){{\footnotesize $n-2$}}} 
		\put(56,-12){\makebox(0,0){{\footnotesize (1,2)}}} 
		\put(84,12){\makebox(0,0){{\footnotesize $n-1$}}} 
		\put(84,-12){\makebox(0,0){{\footnotesize (1,2)}}} 
		\put(112,12){\makebox(0,0){{\footnotesize $n$}}} 
		\put(112,-12){\makebox(0,0){{\footnotesize (2,1)}}} 
	\end{picture}
\qquad \qquad
	$d^{(2)}_{n+1} \ (n\geq2)$ \qquad \qquad
	\begin{picture}(100,40) \thicklines 
	\put(-28,0){\circle{10}}
	\put(-24,3){\line(1,0){20}}\put(-24,-3){\line(1,0){20}} 
		\put(-18,0){\line(1,-1){10}}\put(-18,0){\line(1,1){10}}
		\multiput(0,0)(28,0){5}{\circle{10}} 
		\put(5,0){\line(1,0){18}} 
		\multiput(33,0)(4,0){5}{\line(1,0){2}}
		\put(60,0){\line(1,0){18}} 
		\put(88,3){\line(1,0){20}}\put(88,-3){\line(1,0){20}} 
		\put(101,0){\line(-1,-1){10}}\put(101,0){\line(-1,1){10}}
\put(-28,12){\makebox(0,0){{\footnotesize 0}}} 
		\put(-28,-12){\makebox(0,0){{\footnotesize (1,1)}}} 
		\put(0,12){\makebox(0,0){{\footnotesize 1}}} 
		\put(0,-12){\makebox(0,0){{\footnotesize (2,1)}}} 
		\put(28,12){\makebox(0,0){{\footnotesize 2}}} 
		\put(28,-12){\makebox(0,0){{\footnotesize (2,1)}}} 
		\put(56,12){\makebox(0,0){{\footnotesize $n-2$}}} 
		\put(56,-12){\makebox(0,0){{\footnotesize (2,1)}}} 
		\put(84,12){\makebox(0,0){{\footnotesize $n-1$}}} 
		\put(84,-12){\makebox(0,0){{\footnotesize (2,1)}}} 
		\put(112,12){\makebox(0,0){{\footnotesize $n$}}} 
		\put(112,-12){\makebox(0,0){{\footnotesize (1,1)}}} 
	\end{picture}	
	\\
	\\
	$a^{(2)}_{2n} \ (n\geq2)$ \qquad \qquad
	\begin{picture}(100,40) \thicklines 
	\put(-28,0){\circle{10}}
	\put(-24,3){\line(1,0){20}}\put(-24,-3){\line(1,0){20}} 
		\put(-18,0){\line(1,-1){10}}\put(-18,0){\line(1,1){10}}
		\multiput(0,0)(28,0){5}{\circle{10}} 
		\put(5,0){\line(1,0){18}} 
		\multiput(33,0)(4,0){5}{\line(1,0){2}}
		\put(60,0){\line(1,0){18}} 
		\put(88,3){\line(1,0){20}}\put(88,-3){\line(1,0){20}} 
		\put(93,0){\line(1,1){10}}\put(93,0){\line(1,-1){10}}
\put(-28,12){\makebox(0,0){{\footnotesize 0}}} 
		\put(-28,-12){\makebox(0,0){{\footnotesize (1,2)}}} 
		\put(0,12){\makebox(0,0){{\footnotesize 1}}} 
		\put(0,-12){\makebox(0,0){{\footnotesize (2,2)}}} 
		\put(28,12){\makebox(0,0){{\footnotesize 2}}} 
		\put(28,-12){\makebox(0,0){{\footnotesize (2,2)}}} 
		\put(56,12){\makebox(0,0){{\footnotesize $n-2$}}} 
		\put(56,-12){\makebox(0,0){{\footnotesize (2,2)}}} 
		\put(84,12){\makebox(0,0){{\footnotesize $n-1$}}} 
		\put(84,-12){\makebox(0,0){{\footnotesize (2,2)}}} 
		\put(112,12){\makebox(0,0){{\footnotesize $n$}}} 
		\put(112,-12){\makebox(0,0){{\footnotesize (4,1)}}} 
	\end{picture}

\vspace{-0.5cm}	
$e^{(2)}_6$ \qquad \qquad
	\begin{picture}(100,60) \thicklines 
		\put(-28,0){\circle{10}}
		\put(-24,0){\line(1,0){18}} 	
		\multiput(0,0)(28,0){4}{\circle{10}}
		\put(5,0){\line(1,0){18}} 
		\put(32,-3){\line(1,0){20}}
		\put(32,3){\line(1,0){20}} 
		\put(39,0){\line(1,-1){10}}
		\put(39,0){\line(1,1){10}} 
		\put(61,0){\line(1,0){18}}
		\put(-28,12){\makebox(0,0){{\footnotesize 0}}} 
		\put(0,12){\makebox(0,0){{\footnotesize 1}}} 
		\put(28,12){\makebox(0,0){{\footnotesize 2}}} 
		\put(56,12){\makebox(0,0){{\footnotesize 3}}} 
		\put(84,12){\makebox(0,0){{\footnotesize 4}}} 
		\put(-28,-12){\makebox(0,0){{\footnotesize (1,1)}}} 
		\put(0,-12){\makebox(0,0){{\footnotesize (1,2)}}} 
		\put(28,-12){\makebox(0,0){{\footnotesize (1,3)}}} 
		\put(56,-12){\makebox(0,0){{\footnotesize (2,2)}}} 
		\put(84,-12){\makebox(0,0){{\footnotesize (2,1)}}} 
	\end{picture}
	\qquad \qquad \qquad 
			$f^{(1)}_4$ \qquad \qquad
	\begin{picture}(100,60) \thicklines 
		\put(-28,0){\circle{10}}
		\put(-24,0){\line(1,0){18}} 	
		\multiput(0,0)(28,0){4}{\circle{10}}
		\put(5,0){\line(1,0){18}} 
		\put(32,-3){\line(1,0){20}}
		\put(32,3){\line(1,0){20}} 
		\put(47,0){\line(-1,1){10}}
		\put(47,0){\line(-1,-1){10}} 
		\put(61,0){\line(1,0){18}}
		\put(-28,12){\makebox(0,0){{\footnotesize 0}}} 
		\put(0,12){\makebox(0,0){{\footnotesize 1}}} 
		\put(28,12){\makebox(0,0){{\footnotesize 2}}} 
		\put(56,12){\makebox(0,0){{\footnotesize 3}}} 
		\put(84,12){\makebox(0,0){{\footnotesize 4}}} 
		\put(-28,-12){\makebox(0,0){{\footnotesize (2,1)}}} 
		\put(0,-12){\makebox(0,0){{\footnotesize (2,2)}}} 
		\put(28,-12){\makebox(0,0){{\footnotesize (2,3)}}} 
		\put(56,-12){\makebox(0,0){{\footnotesize (1,4)}}} 
		\put(84,-12){\makebox(0,0){{\footnotesize (1,2)}}} 
	\end{picture}

\vspace{-0.8cm}
$g^{(1)}_2$ \qquad \qquad
	\begin{picture}(30,60) \thicklines 
		\put(-28,0){\circle{10}}	
		\multiput(0,0)(28,0){2}{\circle{10}}
		\put(-24,0){\line(1,0){18}} 
		\put(4,-3){\line(1,0){20}} 
		\put(5,0){\line(1,0){18}}
		\put(4,3){\line(1,0){20}} 
		\put(19,0){\line(-1,-1){10}}
		\put(19,0){\line(-1,1){10}}
		\put(-28,12){\makebox(0,0){{\footnotesize 0}}} 
		\put(0,12){\makebox(0,0){{\footnotesize 1}}} 
		\put(28,12){\makebox(0,0){{\footnotesize 2}}} 
		\put(-28,-12){\makebox(0,0){{\footnotesize (3,1)}}} 
		\put(0,-12){\makebox(0,0){{\footnotesize (3,2)}}} 
		\put(28,-12){\makebox(0,0){{\footnotesize (1,3)}}} 
	\end{picture}
	\qquad \qquad \qquad \qquad \qquad \qquad \qquad $d^{(3)}_4$ \qquad \qquad
	\begin{picture}(30,60) \thicklines 
		\put(-28,0){\circle{10}}	
		\multiput(0,0)(28,0){2}{\circle{10}}
		\put(-24,0){\line(1,0){18}} 
		\put(4,-3){\line(1,0){20}} 
		\put(5,0){\line(1,0){18}}
		\put(4,3){\line(1,0){20}} 
		\put(11,0){\line(1,1){10}}
		\put(11,0){\line(1,-1){10}}
		\put(-28,-12){\makebox(0,0){{\footnotesize(1,1)}}} 
		\put(0,-12){\makebox(0,0){{\footnotesize(1,2)}}} 
		\put(28,-12){\makebox(0,0){{\footnotesize(3,1)}}} 
		\put(-28,12){\makebox(0,0){{\footnotesize 0}}} 
		\put(0,12){\makebox(0,0){{\footnotesize 1}}} 
		\put(28,12){\makebox(0,0){{\footnotesize 2}}} 
	\end{picture}

\vspace{-0.8cm}
$a^{(2)}_2$ \qquad \qquad
	\begin{picture}(30,60) \thicklines 
		\multiput(0,0)(28,0){2}{\circle{10}}
		\put(2,-5){\line(1,0){24}} 
		\put(4,-3){\line(1,0){20}}
		\put(4,3){\line(1,0){20}} 
		\put(2,5){\line(1,0){24}} 
			\put(19,0){\line(-1,-1){10}}
		\put(19,0){\line(-1,1){10}}

		\put(0,12){\makebox(0,0){{\footnotesize 0}}} 
		\put(28,12){\makebox(0,0){{\footnotesize 1}}} 
		\put(0,-12){\makebox(0,0){{\footnotesize(4,1)}}} 
		\put(28,-12){\makebox(0,0){{\footnotesize(1,2)}}} 
	\end{picture}
	\\

\vspace{0.2cm}

\end{document}